\documentclass[12pt]{article}

\usepackage{a4wide}

\usepackage{url}
\usepackage{tikz}
\usepackage[thmmarks]{ntheorem}
\usepackage{newproof}
\usepackage{verbatim}
\usepackage[english]{babel}

\usepackage{color}
\usepackage{latexsym,mathrsfs}
\usepackage{amsfonts, amsmath, amssymb}

\newcommand{\Agents}{{\mathcal N}}
\newcommand{\Candidates}{{\mathcal C}}

\newcommand{\profile}{\succ}
\newcommand{\States}{S}   

\newcommand{\imp}{\rightarrow}
\newcommand{\Imp}{\Rightarrow}
\newcommand{\Pmi}{\Leftarrow}
\renewcommand{\max}{\mathsf{max}}
\renewcommand{\min}{\mathsf{min}}
\newcommand{\eq}{\leftrightarrow}

\newcommand{\et}{\wedge}
\newcommand{\vel}{\vee}
\newcommand{\Et}{\bigwedge}
\newcommand{\Vel}{\bigvee}

\newcommand{\Union}{\bigcup}

\newcommand{\weg}[1]{}
\newcommand{\powerset}{{\mathcal P}}
\newcommand{\pmodel}{{\mathcal M}}
\newcommand{\Inter}{\bigcap}

\newcommand{\vote}{\succ}

\newtheorem{theorem}{Theorem}
\newtheorem{definition}[theorem]{Definition}
\newtheorem{example}[theorem]{Example}
\newtheorem{proposition}[theorem]{Proposition}

\newcommand{\lang}{{\mathcal L}}

\renewcommand{\phi}{\varphi}

\newcommand{\cvote}{[\succ]}

\usepackage{natbib}

\begin{document}

\title{How does uncertainty about other voters determine a strategic vote?}

\author{
Zeinab Bakhtiari\thanks{LORIA, CNRS --- Universit\'e de Lorraine, {\tt bakhtiarizeinab@gmail.com}}
\and 
Hans van Ditmarsch\thanks{LORIA, CNRS --- Universit\'e de Lorraine, {\tt hans.van-ditmarsch@loria.fr}}
\and 
Abdallah Saffidine\thanks{Computer Science \& Engineering, University of New South Wales, {\tt abdallah.saffidine@gmail.com}}
}

\date{}

\maketitle

\begin{abstract} 
We propose a framework for strategic voting when a voter may lack knowledge about  the preferences of other voters, or about other voters' knowledge about her own 
preference. In this setting we define notions of manipulation, equilibrium, and dominance, under uncertainty. We provide scenarios wherein the profiles of true preferences are the same but the equilibrium profiles are different, because the voters have different knowledge about other voters. We also model actions that change such uncertainty about preferences, such as a voter revealing her preference. We show that some forms of manipulation and equilibrium are preserved under such uncertainty updates and others not. We then formalize epistemic voting terminology in a logic. 
Our aim is to provide the epistemic background for the analysis and design of voting rules that incorporate uncertainty.
\end{abstract}

{\bf keywords} \quad social choice, voting, epistemic logic, dynamics

\section{Introduction}

A well-known fact in social choice theory is that strategic voting, also known as manipulation, becomes harder when voters know less about the preferences of other voters. Standard approaches to manipulation in social choice theory \citep{Gibbard73,Satterthwaite75} as well as in computational social choice \citep{BartholdiToveyTrick89b} assume that the manipulating voter knows the sincere preferences of other voters. Some approaches \citep{DugganSchwartz00,BarberaBogomolnaiaStel98} assume that voters have a probabilistic prior belief on the outcome of the vote, which encompasses the case where each voter has a probability distribution over the set of profiles. In yet other approaches the uncertainty of the manipulator is modelled as the inability to distinguish between a set of voting profiles \citep{conitzeretal.aaai:2011,meir:2015}. In \citet{conitzeretal.aaai:2011} manipulators have incomplete knowledge of the non-manipulators' preferences. \citet{meir:2015} is a setting for iterated voting, wherein incomplete knowledge of the profile at the next iteration is induced by the partial view on the profile at the current iteration. Still, we think that the study of strategic voting under complex belief states has received little attention so far, especially when voters are uncertain about the uncertainties of other voters, i.e., when we model higher-order beliefs of voters.


In this contribution we present a logic to model higher-order uncertainty of voters. On the assumption that voters may be uncertain about other voters' preferences but know their own preference, we model how this uncertainty may determine a strategic vote, and how a reduction in this uncertainty may change a strategic vote. We give scenarios where the profile is the same, and even the set of profiles about which the manipulator is uncertain is the same, but where the uncertainty about other voters is different, thus resulting in different manipulations. Additionally, reducing such uncertainty may affect manipulative behaviour, and we also give example scenarios for that. We model uncertainty reducing actions as truthful public announcements \citep{plaza:1989}.

There are several ways of expressing incomplete knowledge about the linear order of preference of a voter. The literature on possible and necessary winners assumes that it is expressed by a collection of partial strict orders (one for each voter), while \citet{HAKW08} consider it to consist of a collection of probability distributions, or a collection of sets of linear orders (one for each voter), i.e., a collection of profiles. For partial preference see also \citet{KonczakLang05} and more generally for voting under incomplete knowledge  see \citet{BoutilierRosenscheinChapter}.

\medskip

A first link between epistemic logic and voting, to our knowledge, has been given in \citet{chopraetal:2004}---they use knowledge graphs to indicate that a voter is uncertain about the preferences of another voter. A more recent follow-up of that, within the area known as social software, is \citet{parikhetal:2013}. 

An independent line of modal logics for social choice, of which voting can be seen as a special case, was proposed in \cite{agotnesetal:2006} and in the journal follow-up \cite{agotnesetalJAAMAS:2011}. They consider two modalities of which one formalizes what is true for all profiles (in the current agenda). It is clearly similar to our epistemic modality formalizing uncertainty about (voting) profiles. The quantification in \cite{agotnesetalJAAMAS:2011} is game theoretical rather than epistemic as in \citet{chopraetal:2004} and in our proposal.

Modal logics of social choice and of voting have further been proposed in \cite{troquardetal:2011} and, building on that, in \cite{cinaetal:2016,perkov:2016}. Also in these logics the semantic primitives are preferences of agents or voters, and sets of those, i.e., profiles. But the modalities encode agency, and not uncertainty, as we do. Compared to their results the logical equivalents of our voting primitives are very `flat': essentially we encode them as big disjunctions of profiles, for example, we represent the proposition that $a$ is the winner of the election by the (very large) disjunction of all profiles wherein this is the case, given that the voting function $F$ is a parameter of the logic (and not of the model, as in \cite{cinaetal:2016}). In our case, all the logical action goes into the uncertainty about profiles, and the modelling results involve the formalization of epistemic notions such as conditional (i.e., Bayesian) equilibrium, and their (epistemic) updates. 

The uncertainty of a manipulating voter in the mentioned \citet{conitzeretal.aaai:2011} and \citet{meir:2015} is modelled in information sets (i.e., set of indistinguishable profiles). However, from that voter's perspective the other voters are not uncertain, so that higher-order uncertainty is not considered. The goals of \citet{meir:2015} are similar to ours (when do equilibria exist, assume risk aversion) but his methods are statistical (there is no higher-order uncertainty). 

Our setting shares also some similarity with robust mechanism design \citep{bergemannetal:2005}, which generalizes classical mechanism design by weakening the common knowledge assumptions of the environment among the players and the planner. In \citet{bergemannetal:2005} uncertainty is modelled with information partitions. The main technical difference is that in our setting, as in classical social choice theory, preferences are ordinal, whereas in (robust) mechanism design preferences are numerical payoffs, which allows for payments. 

\medskip

This is an overview of our contribution. Section \ref{sec.vote} presents voting terminology. Section \ref{sec.knpro} introduces knowledge profiles, our semantic primitive. Section \ref{sec.dom} investigates epistemic notions of dominance. Section \ref{sec.eqkn} defines equilibrium profiles under uncertainty, wherein voting is represented as a Bayesian game. 
Section \ref{sec.dyn} is devoted to uncertainty updates and how this affects knowledge of other voters and equilibrium profiles. Then, in Section \ref{sec.logic} we succinctly present the logic of this contribution more formally --- we do this at the end, because our main focus is the semantic interaction of knowledge and voting, not the logic. 

\section{Voting} \label{sec.vote}


Assume a finite set $\Agents = \{1,\dots,n\}$ of $n$ {\em voters} (or {\em agents}), and a finite set  $\Candidates = \{a,b,c,\dots \}$ of $m$ {\em candidates} (or {\em alternatives}). Voter variables are $i,j,\dots$, and candidate variables are $x,y, (x_1, x_2,) ...$. Let $O(\Candidates)$ be the set of linear orders on $\Candidates$. 
\begin{definition}[Preference, profile, voting rule, vote] For each voter $i$, a {\em preference (relation)} (over $\Candidates$) is a linear order on $\Candidates$. A {\em profile} $\profile$ is a function $\profile: \Agents \imp O(\Candidates)$ assigning a preference $\vote\!\!(i)$, denoted as $\vote_i$, to each voter. A {\em (resolute) voting rule} is a function $F: O(\Candidates)^\Agents \imp \Candidates$ from the set of  profiles for $\Agents$ to the set of candidates. From the perspective of the voting rule the enacted preference $\vote_i$ is a {\em vote}. \end{definition} 
If voter $i$ prefers candidate $a$ to candidate $b$ in preference $\vote_i$, we write $a \succ_i b$, or $b \prec_i a$. For ($a=b$ or $a \succ_i b$) we write $a \succeq_i b$, or $b \preceq_i a$. 
Profile variables are $\profile, \profile', \profile'', ...$. If ${\profile} \in O(\Candidates)^\Agents$ and ${\vote'_i} \in O(\Candidates)$, then $(\profile_{-i},\vote'_i)$ is the profile wherein $\vote_i$ is substituted by $\vote'_i$ in $\profile$.

The voting rule determines which candidate wins the election --- $F(\profile)$ is the {\em winner}.  In case there is more than one {\em tied co-winner} we assume a {\em tie-breaking preference}, that is a linear order over candidates. In the {\em plurality voting rule} the winner is the candidate who is most often ranked the top candidate (most preferred in $\succ_i$), where in case there are several co-winners the tie-breaking preference selects one.

Voters cannot be assumed to vote according to their preference. Instead of giving her {\em sincere} or {\em truthful preference}, a voter may cast another preference as her vote. This is an {\em insincere} or {\em strategic preference}. If that vote improves the outcome it is a {\em manipulation}.

\begin{definition}[Manipulation] \label{def.success} Let $i \in \Agents$, ${\profile} \in O(\Candidates)^\Agents$, ${\vote'_i} \in O(\Candidates)$, and $F$ a voting rule. If $F(\profile_{-i},\vote'_i)$ $\succ_i F(\profile)$, then $\vote'_i$ is a {\em manipulation} by voter $i$ of profile $\profile$.
\end{definition}

The combination of a profile $\profile$ and a voting rule $F$ defines a strategic game: a player is a voter, an individual strategy for a player is a preference, a strategy profile (of players) is therefore a profile in our defined sense (of voters), and the preference of a player among the outcomes is according to her sincere vote: given profiles $\profile'$, $\profile''$, voter $i$ also prefers outcome $F(\profile')$ over outcome $F(\profile'')$ in the game theoretical sense iff (in the voting sense) $F(\profile') \succ_i F(\profile'')$. The relevant notions of dominance and equilibrium are as follows.

\begin{definition}[Dominant preference] \label{def.dominantvote} Let $i \in \Agents$,   ${\profile} \in O(\Candidates)^\Agents$, ${\vote_i'} \in O(\Candidates)$, and $F$ a voting rule. If for all ${\profile''} \in O(\Candidates)^\Agents$, $F(\profile''_{-i},\vote'_i) \succeq_i F(\profile'')$, and for some ${\profile''} \in O(\Candidates)^\Agents$, $F(\profile''_{-i},\vote'_i) \succ_i F(\profile'')$, then $\vote'_i$ is a {\em dominant preference} for voter $i$.
\end{definition}
A dominant preference corresponds to a dominant strategy in game theory.\footnote{We called this (game theoretical) notion `dominant preference' and not `dominant manipulation', to avoid confusion with the (epistemic) notion of {\em dominant manipulation of an information set}, defined in Subsection \ref{sec.dominance}: only the latter is commonly used in voting theory, as the requirements of the former are very strong.} 
The reader will recognize this as {\em weak} dominance. We may also use {\em strong} dominance, which holds for $\succ_i'$ if for all ${\profile''} \in O(\Candidates)^\Agents$, $F(\profile''_{-i},\vote'_i) \succ_i F(\profile'')$. Note that a dominant preference need not be a manipulation of $\succ$: the strict part may be satisfied for another profile than the profile $\profile$ of true preferences. (Of course, a {\em strongly} dominant strategy is also a manipulation.)

\begin{definition}[Equilibrium profile] \label{def.eq}
A profile $\profile$ is an {\em equilibrium profile} iff no voter has a manipulation of $\profile$.
\end{definition} 
An equilibrium profile corresponds to a Nash equilibrium (in game theory). An equivalent way of defining equilibrium profile is: A profile $\profile$ is an equilibrium profile iff for all ${\vote'_i} \in O(\Candidates)$, $F(\profile) \succeq_i F(\profile_{-i},\vote'_i)$. The formulations correspond: 

Suppose the above condition does not hold. Then there is a voter $i$ and some preference ${\vote'_i} \in O(\Candidates)$ such that $F(\profile) \prec_i F(\profile_{-i},\vote'_i)$, i.e., $F(\profile_{-i},\vote'_i) \succ_i F(\profile)$, i.e., $\vote_i'$ is a manipulation for voter $i$ of profile $\profile$. In other words, if the condition holds, then no voter has a manipulation of $\profile$.

\section{Knowledge profiles} \label{sec.knpro}

We model uncertainty in voting as incomplete knowledge about the profile. The structures to represent such uncertainty are standard in modal logic \citep{faginetal:1995,hvdetal.del:2007}. To allow for the definition of dominance and of equilibria under uncertainty, we require that voters know their own preference.

\begin{definition}[Knowledge profile] 
A {\em profile model} is a structure $\pmodel = (S, \sim, \pi)$, where $S$ is a {\em domain} of abstract objects called states; where ${\sim} : \Agents \imp \powerset(S\times S)$ is a function such that for $i = 1,\dots,n$, ($\sim\!\!(i)$ written as) $\sim_i$ is an {\em indistinguishability relation} that is an equivalence relation; and where {\em valuation} $\pi: \States \imp O(\Candidates)^\Agents$ assigns a profile to each state and such that $s \sim_i t$ implies $\pi(s)_i = \pi(t)_i$. The {\em information set} of voter $i$ in state $s$ is defined as $[s]_{\sim_i} := \{ t \mid s \sim_i t \}$. Let $\pi([s]_{\sim_i})$ denote $\{ \pi(t) \mid s \sim_i t \}$. A {\em knowledge profile} is pointed structure $\pmodel_s$ where $s \in \States$. \end{definition} 
Unless confusion results, the set $\pi([s]_{\sim_i})$ of profiles that $i$ considers possible is also called an information set, as in the voting literature  \citep{conitzeretal.aaai:2011}. In general, different information sets (i.e., sets of states) may be about the same set of profiles.
\begin{definition}[Knowledge and ignorance]
Given is a knowledge profile $\pmodel_s$, where $\pmodel = (S, \sim, \pi)$. Let $\phi$ be a proposition about profiles. Voter $i$ {\em knows} $\phi$ in $\pmodel_s$, iff $\phi$ is true in all $t \in S$ such that $s \sim_i t$. Voter $i$ {\em considers possible that} (or {\em does not know that not}) $\phi$ in $\pmodel_s$, iff $\phi$ is true in some $t \in S$ such that $s \sim_i t$; if, in that case, there is an additional state $u \in S$ with $s \sim_i u$ in which $\phi$ is false, then we say that $i$ {\em does not know whether} (or {\em is uncertain about}, or {\em is ignorant about}) $\phi$.
\end{definition}
Section \ref{sec.logic} contains a proper inductive definition of `proposition about profiles', and a formal semantics.  For now it suffices to say that the following are propositions about profiles: $\profile$ or `the profile is $\profile$'; $a \succ_i b$, and $\vote_i$ (`the preference of voter $i$ is $\vote_i$'); and also propositions like `voter $i$ knows that voter $j$ knows $\vote_i$'. 
Example \ref{ex.one} demonstrates that: {\em different states may be assigned the same profile, but have different knowledge properties}. In scenarios where different states are always assigned different profiles, we can say that the uncertainty of a voter is (only) about a collection of profiles. But in scenarios where different states are assigned the same profile, the set $\pi([s]_{\sim_i})$ of profiles that voter $i$ considers possible is smaller than the set $[s]_{\sim_i}$ of states that $i$ considers possible.



Partial preferences cannot be expressed in our framework. In particular, uncertainty between $a \succ_i b \succ_i c$ and $b \succ_i a \succ_i c$ does not mean indifference between candidates $a$ and $b$. Uncertainty between $a \succ_i b \succ_i c$ and $b \succ_i a \succ_i c$ means that ($a \succ_i b \succ_i c$ or $b \succ_i a \succ_i c$) is true. This entails ($a \succ_i b$ or $b \succ_i a$), which is equivalent to (not ($b \succeq_i a$ and $a \succeq_i b$)). That is the opposite of indifference between $a$ and $b$, as that means ($b \succeq_i a$ and $a \succeq_i b$).

\begin{example} \label{ex.one}
Consider two voters who are Leela ($1$) and Sunil ($2$), children of Devi, and who `vote' for an animated movie to see before bedtime; where the choice is between $a$ (Alice in Wonderland), $b$ (Brave), and $c$ (Cars). Leela's preference is  $a \succ_1 b \succ_1 c$ and Sunil's preference is $c \succ_2 b \succ_2 a$.  Sunil (2) is uncertain if Alice is Leela's most or least preferred movie, and dually if Cars is Leela's least or most preferred movie. More interestingly, Leela (1) knows Sunil's preference, but she is uncertain whether Sunil knows her preference. Even more interestingly, Sunil is also uncertain whether, in case Alice is Leela's most preferred movie, she knows that he does not know that.

We model this as a knowledge profile $\pmodel_t$ consisting of three states $s,t,u$ and for two voters 1 and 2. State $s$ is assigned to profile $\profile$, wherein $a \succ_1 b \succ_1 c$ and $c \succ_2 b \succ_2 a$, etc. Below, a column represents a preference relation, and states that are indistinguishable for a voter $i$ are linked with an $i$-labelled edge. The partition for 1 on the domain is therefore $\{\{ s,t\}, \{u\} \}$, and the partition for 2 on the domain is $\{\{ s\}, \{t,u\}\}$.

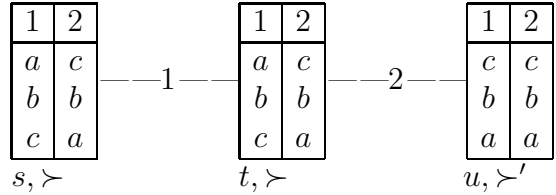
\[\begin{array}{l}
\begin{array}{|c|c|} \hline 1 & 2 \\ \hline a & c \\ b & b \\ c & a \\ \hline \end{array}\mbox{------}1\mbox{------}\begin{array}{|c|c|} \hline 1 & 2 \\ \hline a & c \\ b & b \\ c & a \\ \hline \end{array}\mbox{------}2\mbox{------}\begin{array}{|c|c|} 
\hline 
1 & 2 \\ 
\hline 
c & c \\ 
b & b \\ 
a & a \\ 
\hline 
\end{array}
\\
\hspace{.0cm} s, \profile \hspace{2.2cm} t, \profile \hspace{2.2cm} u, \profile' \end{array}\]

States $s$ and $t$ have been assigned the same profile $\profile$ but have different epistemic properties. In $s$, 2 knows that 1 prefers $a$ over $c$, whereas in $t$, 2 does not know that. We list some relevant propositions that are true in the actual state $t$: 
\begin{itemize}
\item Leela prefers Alice over Cars. 
{\em This is true, because $a \succ_1 c$ in $t$.}
\item Sunil does not know that Leela prefers Alice over Cars. 
{\em This is true, because $t\sim_2u$, and $a \succ_1 c$ is false in $u$.}
\item Leela knows Sunil's preference, but she is uncertain whether Sunil knows her preference. {\em In $s$, Sunil knows that Leela's preference is $a \succ_1 b \succ_1 c$, whereas in $t$, Sunil does not know that Leela's preference is $a \succ_1 b \succ_1 c$, because $t\sim_2u$, and $c \succ_1 b \succ_1 a$ in $u$.}
\end{itemize}
\end{example}

Unlike merely sincere and insincere preference, in knowledge profiles there are {\em three} kinds of preference: actual sincere preference, possible sincere preference, and insincere preference. In $\pmodel_t$ of Example \ref{ex.one}, Leela's actual sincere preference is $a \succ_1 b \succ_1 c$, a possible sincere preference is $c \succ_1 b \succ_1 a$ (namely in state $u$, from the perspective of Sunil), and an insincere preference is $b \succ_1 a \succ_1 c$. This can be confusing.

\section{Manipulation, knowledge and dominance} \label{sec.dom}


\subsection{Manipulation and knowledge}

Given are a knowledge profile $\pmodel_s$ where $\pi(s) = {\profile}$, and a voting rule $F$. If voter $i$ can manipulate $\profile$, then voter $i$ can also manipulate $\pmodel_s$. This is because manipulation is defined with respect to the profile of the actual state of the knowledge profile (it is a game theoretical notion, not an epistemic notion). So it may be that a voter can manipulate the vote but does not know that, because she considers another profile possible wherein she cannot manipulate the vote. Notions of manipulation that involve knowledge can be borrowed from the knowledge and action literature \citep{jfak.bulletin:2001,Jamroga2004}. A curious situation is when in all states that the voter considers possible there is a manipulation, but when in different such states there are different manipulations. So she knows that she has a manipulation, but she does not know what the manipulation is. This is called {\em de dicto knowledge} of manipulation. A stronger form of knowledge is when there is a preference $\vote'_i$ that is {\em the same} manipulation in any state that the voter considers possible. This is called {\em de re knowledge} of manipulation.  
\begin{definition}[Knowledge of manipulation] \label{def.knowmanip} Given are a knowledge profile $\pmodel_s$ and a voting rule $F$.
\begin{itemize}
\item Voter $i$ knows {\em de dicto}  that she can manipulate $\pmodel_s$, if for all profiles  ${\profile}\in\pi([s]_{\sim_i})$ there is a preference $\vote'_i$ such that $\vote'_i$ is a manipulation in $\profile$.
\item Voter $i$ knows {\em de re} that she can manipulate $\pmodel_s$, if there is a preference $\vote'_i$ such that for all profiles  ${\profile}\in\pi([s]_{\sim_i})$, $\vote'_i$ is a manipulation in $\profile$.
\end{itemize}
\end{definition}
If voter $i$ knows {\em de re} that she can manipulate the election, she has the ability to manipulate, namely by strategically voting $\vote'_i$. But in {\em de dicto} manipulations the voter does not seem to have that ability. It is akin to `game of chicken' type equilibria in game theory, wherein for each strategy of a player there is a complementary strategy of the other player such that the pair is an equilibrium, but where this choice cannot be made without coordination. An example of de dicto knowledge of manipulation for Borda voting is given in \cite{hvdetal.TARKvote:2013}. 

Consider the profile model $\mathcal{H}$ consisting of the domain $O(\Candidates)^\Agents$, so we can identify states $s$ with their profiles ${\profile} = \pi(s)$,  and such that all voters only know their own preferences: ${\profile} \sim_i {\profile'}$ iff ${\vote_i} = {\vote'_i}$. In this (unique) model it is common knowledge that voters only know their own preferences. We can see it as an {\em interpreted system} \citep{faginetal:1995} consisting of {\em global states} that are profiles and where {\em local states} are individual preferences. A model such as $\mathcal{H}$ is known as a {\em hypercube} \citep{lomuscio:1999}.
\begin{proposition} \label{prop.impossible}
In the hypercube profile model, knowledge of manipulation is impossible for plurality voting.
\end{proposition}

\begin{proof}
The result holds for {\em de re} knowledge of manipulation and for {\em de dicto} knowledge of manipulation. We start with the {\em de re} case. 

Let $\profile$ be the profile. Let us assume that there are a sufficient number of voters and candidates to avoid boundary cases. Assume towards a contradiction that voter $i$ knows that $\vote'_i$ is a manipulation. As $i$ is uncertain about the preferences of other voters, she considers it possible all other voters $j$ have the same preferences as herself, i.e., she considers it possible that the profile is ${\profile''}\in O(\Candidates)^\Agents$ such that for all $j \in \Agents$, ${\vote_j''} = {\vote_i}$. In that case, $i$'s preferred candidate would have won by majority vote, contradicting the assumption that $\vote'_i$ is a manipulation. 

In the {\em de dicto} case, for each profile that voter $i$ considers possible there is a manipulation. For all those profiles, assuming that all other voters $j$ have the same preference as $i$, again derives a contradiction.
\end{proof}

\subsection{Dominant manipulation and knowledge} \label{sec.dominance}

We now compare the notions of manipulation and dominant strategy with the  notion of {\em dominant manipulation of an information set} in voting theory \citep{conitzeretal.aaai:2011}.
\begin{definition}[Dominant manipulation of an information set] 
Let a knowledge profile $\pmodel_s$ with $\pi(s) = {\profile}$, $i \in \Agents$, ${\vote_i'} \in O(\Candidates)$, and a voting rule $F$ be given. If for all ${\profile''} \in \pi([s]_{\sim_i})$, $F(\profile''_{-i},\vote'_i)$ $\succeq_i F(\profile'')$, and for some ${\profile''} \in \pi([s]_{\sim_i})$, $F(\profile''_{-i},\vote'_i) \succ_i F(\profile'')$, then $\vote'_i$ is a {\em dominant manipulation} for voter $i$ of information set $\pi([s]_{\sim_i})$ (or: of knowledge profile $\pmodel_s$).
\end{definition}
Observe that dominant manipulation of an information set  according to \cite{conitzeretal.aaai:2011} is on the assumption that all other voters vote according to their true preference, whereas dominant preference in the game theoretical sense (Def.\ \ref{def.dominantvote}) is on the assumption that all other voters can choose any preference as their vote. The first is dominance {\bf no matter the true preference} of others (but assuming that this is their vote), the second is dominance {\bf no matter the vote} of others (but assuming what their true preferences are). We consider this difference curious. However, despite such seemingly orthogonal epistemic and game-theoretical dimensions, they are after all very much related, as now shown in the following proposition. We therefore find the observations made in this proposition, although elementary, somewhat surprising. 
\begin{proposition}   \label{fact.ign} \ 
\begin{enumerate}
\item \label{fone} A dominant manipulation of a singleton set $\{\profile\}$ is a manipulation of $\profile$ (Def.\ \ref{def.success}).
\item \label{ftwo} A dominant manipulation of the hypercube knowledge profile $\mathcal{H}_s$ is a dominant preference given profile $\pi(s)$ (Def.\ \ref{def.dominantvote}).
\item \label{fthree} If a voter has a dominant manipulation then she \emph{knows} that she has a dominant manipulation.
\item \label{fxx} If a voter has a dominant manipulation then she may not have a manipulation.
\item Knowledge of dominant manipulation does not imply knowledge of manipulation (neither {\em de re} nor {\em de dicto}) (Def.\ \ref{def.knowmanip}).
\item \label{ffour} If a voter has {\em de re} knowledge of manipulation then she has a dominant manipulation.
\end{enumerate}
\end{proposition}
\begin{proof} \ 
\begin{enumerate}
\item The strictness requirement of dominance must apply to $\profile$.
\item See the curious observation above. We recall that the notion of \emph{dominant manipulation} models uncertainty over the true preferences of others (the `epistemic dimension'), whereas the notion of \emph{dominant preference} models uncertainty over how others vote (the `game-theoretical dimension'). That the two coincide is a result.
\item The notion is defined with respect to an information set.
\item The strictness requirement of dominant manipulation might apply to another than the actual profile.
\item The strictness requirement of dominant manipulation need not apply to all profiles in the information set, as in knowledge of manipulation (both {\em de re} and  {\em de dicto}).
\item Strictness holds for all profiles in the information set, and therefore for some.
\end{enumerate} \vspace{-.65cm}
\end{proof}
Given Proposition \ref{fact.ign}.\ref{ftwo}, Proposition \ref{prop.impossible} stating that knowledge of manipulation is impossible in the hypercube, for majority voting, should therefore be credited to \cite{conitzeretal.aaai:2011} who proves that dominant manipulation is impossible under common ignorance of others' preferences, for a variety of voting rules (their results were subsequently strengthened in \cite{ReijngoudE12}). Proposition \ref{fact.ign}.\ref{ftwo} is somewhat surprising: it says that whether a preference is dominant does not depend on what you know of other voters' preferences. It holds given common knowledge of the profile iff it holds given common ignorance (except one's own preference) of the profile. Dominant strategy as in Def.\ \ref{def.dominantvote} seems too strong to be useful in voting theory, as it does not even rule out that everybody except you acts against their interests.

In view of Proposition \ref{fact.ign}.\ref{ffour}, an alternative designation for  a ({\em de re}) known manipulation is {\em strongly dominant manipulation}.

\section{Equilibrium and knowledge} \label{sec.eqkn}

\subsection{Conditional equilibrium} 

Determining equilibria under incomplete knowledge comes down to decision taking under incomplete knowledge. Therefore we have to choose a decision criterion. Expected utility makes no sense here, because we didn't start with probabilities over profiles in the first place, nor with utilities. In the absence of prior probabilities, the following three criteria make sense. $(i)$ The {\em insufficient reason} (or {\em Laplace}) criterion considers all possible states in a given situation as equiprobable. This criterion was used in \citet{agotnesetal:2011} to determine equilibria of certain (Bayesian) games of imperfect information. $(ii)$ The {\em minimax regret} criterion selects the decision minimizing the maximum utility loss, taken over all possible states, compared to the best decision, had the voter known the true state. $(iii)$ The {\em pessimistic} (or {\em Wald}, or {\em maximin}) criterion compares decisions according to their worst possible consequences. The latter criterion, that we also call {\em risk averse}, is one that fits well our probability-free and utility-free model; this was also the criterion chosen in \citet{conitzeretal.aaai:2011,meir:2015}. The only assumption here is that the probability distribution is positive in all states. We now fix this criterion for the rest of the paper. Pessimistic, optimistic, and other criteria only assuming positive probability are applied to social choice settings in \citet{parikhetal:2013}. \citet{meir:2015} also considers the minimax-regret criterion.

In the presence of knowledge, and on the assumption that voters know their own preference (so that, in game theoretical terms, the payoff function is uniform throughout an agent's information set),  
the definition of an equilibrium extends naturally. For each agent, the combination of an agent $i$ and an information set $[s]_{\sim_i}$ for that agent (for some state $s$ in the knowledge profile) defines a so-called virtual agent: we model these imperfect information games as Bayesian games \citep{harsanyi:1967}. Each virtual agent has the same set of strategies as the `original' agent. An equilibrium is then a profile of strategies such that none of the virtual agents has an interest to deviate. An alternative way to present a Bayesian game, applied in \citet{agotnesetal:2011}, is to change the set of strategies instead of the set of agents. Instead of each agent in each information set (a `virtual agent') choosing a strategy among the set of strategies, we have each agent choosing a {\em conditional strategy} among the larger set of conditional strategies, where conditions correspond to the information sets. We also follow that presentation for voting.  

For risk-averse voters knowing their own preferences we can effectively determine if a conditional profile is an equilibrium without taking probability distributions into account, unlike in the more general setting of Bayesian games that it originates with. 

For any $\Candidates' \subseteq \Candidates$, $\min_i \, \Candidates'$ is the (unique) $c \in \Candidates'$ such that $c' \succeq_i c$ for all $c' \in \Candidates'$. Let $\mathcal{P}$ be a set of profiles, ${\vote'_i} \in O(\Candidates)$, and $F$ a voting rule, then $\min_i \, F(\mathcal{P})$ denotes $\min_i \{ F(\profile) \mid {\profile} \in \mathcal{P} \}$, and $\min_i \, F(\mathcal{P}_{-i}, \vote'_i)$ denotes $\min_i \, \{ F(\profile_{-i},\vote'_i) \mid {\profile} \in \mathcal{P} \}$.

\begin{definition}[Pessimistic manipulation] Given is a profile model $\pmodel = (\States, \sim,\pi)$, $s \in \States$ with $\pi(s) = {\profile}$, and voting rule $F$.  The worst outcome for voter $i$ in information set $\pi([s]_{\sim_i})$ is $\min_i \, F(\pi([s]_{\sim_i}) )$. Preference ${\vote_i'} \in O(\Candidates)$ is a {\em pessimistic manipulation} for voter $i$ of $\pi([s]_{\sim_i})$ (or: of knowledge profile $\pmodel_s$) iff \[ \min_i \, F(\pi([s]_{\sim_i})_{-i},\vote'_i) \ \succ_i \ \min_i \, F(\pi([s]_{\sim_i})). \]
\end{definition}

\begin{definition}[Conditional preference, conditional profile, conditional equilibrium] \label{def.condeq} 
Given is a profile model $\pmodel = (\States, \sim,\pi)$ and voting rule $F$.\weg{such that every voter knows her preferences (truthful vote).} For each voter $i$, a {\em conditional preference} is a function $\cvote_i: \States/{\sim_i} \imp O(\Candidates)$ that assigns a preference to each information set for that voter. A {\em conditional profile} is a function from voters to conditional preferences. A conditional profile is a {\em conditional equilibrium} iff no agent has a pessimistic manipulation of any of its information sets.\footnote{If all states are considered equiprobable (the Laplace criterion), a sufficient (but not necessary) condition for a conditional profile to be a conditional equilibrium is that {\em no agent has a {\bf dominant} manipulation of any of its information sets}. We have not investigated this further.}
\end{definition}
In the situation without uncertainty, given $n$ voters, a profile and a voting rule determine a winner. In the strategic game matrix, the outcome is the $n$-tuple of values (payoffs) of that winner for each voter, and to determine if it is an equilibrium we compare the value for any voter $i$ with the value when $i$ had voted differently: the value should not be higher. The outcome of a conditional profile is not an $n$-tuple of values, but an $n$-tuple of $m$-tuples or vectors $(x_i^1,\dots,x_i^m)$, where voter $i$ has $m$ information sets and where $x_i^1,\dots,x_i^m$ are expected outcomes. These vectors are unordered, so we have to compute equilibria differently. For example, given a voter 7 with two information sets $x$ and $y$, we cannot say which of payoffs $(0,1)$ and $(1,0)$ she prefers. But we can say that virtual voter $(7,x)$ prefers the second (wherein she gets 1) over the first (wherein she gets 0), and that virtual voter $(7,y)$ prefers the first over the second. This merely is the Bayesian game calculation of equilibrium for virtual agents. 

A notable fact, that we consider a main result of our contribution, is that:

\begin{proposition} \label{prop.samedifferent}
States with the same profile can have different conditional equilibria.
\end{proposition}
\begin{proof}
We prove this by example, in the next subsection. The reader may wish to verify in Figure \ref{fig.tuvexample} in that subsection: that Sunil (2) has the same preference in state $u$ as in state $v$, that voting for Cars ($c$) is not in a conditional equilibrium in $u$, whereas voting for Cars is in a conditional equilibrium in $v$.
\end{proof}
In other words, Proposition \ref{prop.samedifferent} states that even if all voters have the same preferences, then when their knowledge about others' preferences is different, their manipulative behaviour may also be different. Readers who find this result obvious may wish to skip the next subsection and proceed with Section \ref{sec.dyn} on revealing voting preferences.

\subsection{Examples of conditional equilibria in plurality voting} \label{sec.ex1}

We recall Example \ref{ex.one} about Leela ($1$) and Sunil ($2$) voting, by plurality, for an animated movie that may be $a$ (Alice), $b$ (Brave), or $c$ (Cars), where Leela's preference is  $a \succ_1 b \succ_1 c$ and Sunil's preference is $c \succ_2 b \succ_2 a$. We further assume that mother Devi, the central authority, has tie-breaking preference $b\succ_\mathsf{tie} a\succ_\mathsf{tie} c$. We present equilibria when there is: no uncertainty, uncertainty between two states with different profiles, and different kinds of uncertainty between three states (for two profiles). 

\paragraph*{\bf No uncertainty} 
We express the payoffs for both voters by their ranking (0, 1, or 2) for the winner. As this is majority voting, preference relations with the same most preferred candidate are indistinguishable. So, `Leela votes $a \succ_1 b \succ_1 c$' and `Leela votes $a \succ_1 c \succ_1 b$' can both be represented by `Leela votes $a$'. (Given this identification, we call $a$ a vote and not a preference (relation).) This simplifies the outcomes matrix and the payoff matrix. If 1 votes for her preference $a$ and $2$ votes for his preference $c$, then the tie-breaking preferences determines $a$ as the winner, 2's least preferred candidate. A strategic vote of 2 for candidate $b$ makes $b$ win, a better outcome for voter 2. Equilibrium pairs of votes are $(a,b)$ and $(b,b)$. For voter 1, voting $a$ is dominant. 

\begin{figure}[h]
\[
\begin{array}{l}
	\begin{array}{|c|c|} \hline 1 & 2 \\ \hline a & c\\ b & b \\ c & a\\ \hline \end{array} \hspace{.7cm} \begin{array}{|c|c|} 
	\hline 
	1 & 2 \\ 
	\hline 
 c & c \\ 
	b & b \\ 
	a & a \\ 
	\hline 
	\end{array}
\hspace{.7cm}
\begin{array}{c|ccc} 
		1\backslash 2 & a & b & c  \\ 
		\hline
		a   & a & b & a  \\ 
		b   & b & b & b \\ 
		c   & a & b & c \\ 
		
	\end{array}
	\hspace{.7cm}
	\begin{array}{c|ccc} 
		1\backslash 2 & a & b & c  \\ 
		\hline
		a   & 2.0 & \fbox{1.1} & 2.0 \\ 
		b   & 1.1 & \fbox{1.1} & 1.1\\ 
		c   & 2.0 & 1.1 & 0.2 \\ 
		
	\end{array}
	\hspace{0.7cm}
		\begin{array}{c|ccc} 
			1\backslash 2 & a & b & c  \\ 
			\hline
			a   & 0.0 & \fbox{1.1} & 0.0  \\ 
			b   & \fbox{1.1} &\fbox{1.1}& 1.1 \\ 
			c   & 0.0 & 1.1 & \fbox{2.2}\\ 
			\end{array} 
\\
\ \ \profile \hspace{1.4cm} \profile' 
\end{array}
\]
\caption{Equilibria for $\profile$ and for $\profile'$. From left to right: profile $\profile$, profile $\profile'$, the matrix with winners, the outcome matrix for profile $\profile$, and the outcome matrix for profile $\profile'$. Equilibria are boxed. We write $i.j$ instead of $(i,j)$ to denote the values for voters 1 and 2 of the outcome of the election.}
\label{fig.zero}
\end{figure}

The other profile used in the examples in this section is where 1 shares the preferences of 2. This is the profile $\profile'$. 
Although $(c,c)$ is an equilibrium vote for $\profile'$, there are various suboptimal equilibria. There is no dominant vote. An overview of the equilibria for $\profile$ and for $\profile'$ is in Figure \ref{fig.zero}. 

\paragraph*{\bf Uncertainty between two profiles} 

Now consider the profile model consisting of two states $t$ with profile $\profile$ and $u$ with profile $\profile'$ and that are the same for voter 2, but different for voter 1 (the accessibility relation for voter 1 is the identity on the model and for voter 2 it is the universal relation). Figure \ref{fig.tuexample} depicts that profile model, the strategic game matrix with conditional preferences and winners, and the strategic game matrix with payoffs. Conditional profiles are pairs $(ij,k)$ where $i$ is 1's vote in $t$ and $j$ is 1's vote in $u$, and $k$ is 2's vote in $\{t,u\}$.

As voter 1 has two information sets, the conditional preference for 1 has two conditions, for each of which a choice between the three candidates (co-)determines the outcome of the majority vote. There are therefore 9 conditional preferences for voter 1. The matrix shows the conditional preferences for 1 by the candidate she votes for in $t$, followed by the candidate she votes for in $u$. Conditional preference $xy$ for 1 means that in state $t$ 1 votes $x$ and in state $u$ 1 votes $y$. The payoff matrix next to the winners matrix contains triples $ij.k$ for, in this order: the value of the worst outcome for 1 given $t$ ($\profile$) of that conditional profile, the value of the outcome for 1 given $u$ ($\profile'$), and the value of the worst outcome for 2 given $t,u$, his only information set. 

For example, for conditional profile $(ba,c)$ we get $ba$ as the entry in the winners matrix and $(10.0)$ as the entry in the payoff matrix: if 1 votes $b$ and 2 votes $c$ then the tie ($b \succ_\mathsf{tie} a \succ_\mathsf{tie} c$) makes $b$ win, value 1 for voter 1 and value 1 for voter 2 in $\profile$; if 1 votes $a$ and 2 votes $c$ then $a$ wins, value 0 for voter 1 and for voter 2 in $\profile'$; the worst of 0 and 1 is 0, so the value for voter 2 of this conditional profile is 0.

\begin{figure}[h]
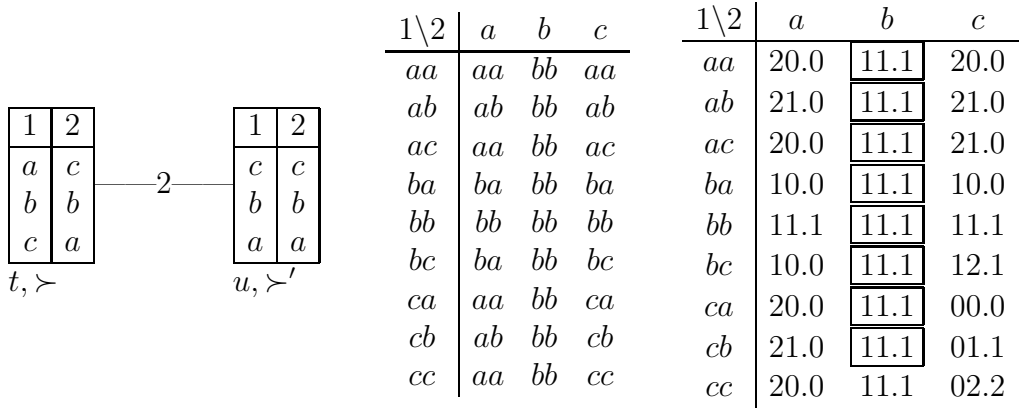

\[
\begin{array}{l}
 \begin{array}{|c|c|} \hline 1 & 2 \\ \hline a & c\\ b & b \\ c & a\\ \hline \end{array}\mbox{------}2\mbox{------}\begin{array}{|c|c|} 
	\hline 
	1 & 2 \\ 
	\hline 
 c & c \\ 
	b & b \\ 
	a & a \\ 
	\hline 
	\end{array}
\\ 
t, \profile \hspace{2.2cm} u, \profile'
\end{array}
\hspace{.7cm}
\begin{array}{c|ccc} 
1\backslash 2 & a & b & c \\ 
\hline
aa   & aa &  bb & aa  \\ 
ab   & ab & bb & ab \\ 
ac   & aa & bb& ac \\ 
ba   & ba & bb& ba \\ 
bb   & bb & bb & bb \\ 
bc   & ba & bb& bc \\ 
ca   & aa & bb & ca \\ 
cb   & ab & bb& cb \\ 
cc   & aa & bb& cc\end{array}
\hspace{.7cm}
\begin{array}{c|ccc} 
1\backslash 2 & a & b & c \\ 
\hline
aa   & 20.0 & \fbox{11.1} & 20.0 \\ 
ab   & 21.0 &\fbox{11.1} & 21.0  \\ 
ac   & 20.0 &\fbox{11.1}& 21.0 \\ 

ba   & 10.0 & \fbox{11.1}& 10.0  \\ 
bb   & 11.1 & \fbox{11.1} & 11.1 \\ 
bc   & 10.0 & \fbox{11.1}& 12.1 \\ 

ca   & 20.0 & \fbox{11.1} & 00.0 \\ 
cb   & 21.0 & \fbox{11.1}& 01.1  \\ 
cc   & 20.0 &  11.1& 02.2
\end{array}
\]
\caption{Voter 2 is uncertain whether voter 1 prefers $a$ over $c$ or $c$ over $a$}
\label{fig.tuexample}
\end{figure}

The equilibria are, maybe, as expected. (We only consider pure strategies.) If the profile is $\profile$ then it is still dominant for voter 1 to vote $a$ (if the profile is $\profile'$, voting for $c$ is not dominant for voter 1). Because voter 2 is risk averse, $(c,c)$ is no longer an equilibrium vote in $\profile'$. As 2 is uncertain whether 1 prefers $c$ over $a$ or $a$ over $c$, the safer (risk avoiding) strategy for 2 is now to vote $b$, even when 1 and 2 both prefer $c$. Voter 1 knows this as well. 

Voter 2 does not have a dominant preference, because if he assumes that voter 1 always votes $c$, the best response is also to vote $c$ and not to vote $b$. So this is the only case where voting $b$ is not an equilibrium vote for 2.

\paragraph*{\bf Uncertainty between three states} 

We now add further uncertainty to the two-state profile model where 2 is uncertain between profiles $\profile$ and $\profile'$. Figure \ref{fig.tuvexample} displays two different ways to do this. In both depicted profile models voter 1 always knows voter 2's preferences. We will show that it is not rational for 2 to behave (vote) differently in $s$ and in $t$, in the first, but that it is rational for 2 to behave differently in $u$ and in $v$, in the second. 


\begin{figure}
{\small
\[ 
\begin{array}{l}
\begin{array}{|c|c|} \hline 1 & 2 \\ \hline a &  c\\ b& b \\ c & a  \\ \hline \end{array}\mbox{------}1\mbox{------}\begin{array}{|c|c|} \hline 1 & 2 \\ \hline a& c \\ b & b \\ c & a \\ \hline \end{array}\mbox{------}2\mbox{------}\begin{array}{|c|c|} 
	\hline 
	1 & 2 \\ 
	\hline 
	c & c \\ 
	b & b \\ 
	a & a \\ 
	\hline 
	\end{array} 
\\
\hspace{.0cm} s, \profile \hspace{2.2cm} t, \profile \hspace{2.2cm} u, \profile'
\end{array}
\hspace{.7cm}
\begin{array}{c|ccccccccc} 
1\backslash 2 & aa & ab & ac & ba & bb & bc & ca & cb & cc \\ 
\hline
aa   & aaa & abb & aaa & baa & bbb & baa & aaa & abb & aaa \\ 
ab   & aab& abb & aab & bab & bbb & bab & aab & abb & aab\\
ac & aaa & abb& aac& baa&  bbb & bac & aaa & abb & aac\\ 
ba   & bba & bbb & bba  & bba  & bbb & bba & bba& bbb& bba\\ 
bb & bbb & bbb & bbb& bbb & bbb & bbb & bbb & bbb & bbb \\
bc & bba & bbb & bbc  & bba  & bbb & bbc& bba& bbb& bbc\\
ca & aaa & abb & aca & baa & bbb & bca  & caa & cbb & cca\\
cb & aab & abb & acb & bab& bbb& bcb & cab& cbb& ccb\\
cc & aaa & abb & acc  & baa& bbb& bcc & caa& cbb& ccc				
\end{array} \] \[ \begin{array}{c|ccccccccc} 
1\backslash 2 & aa & ab & ac & ba & bb & bc & ca & cb & cc \\ 
\hline
aa   & 20.00 & 11.01 & 20.00 & 10.10 & \fbox{11.11}& 10.10  & 20.00 & 11.01& 20.00 \\ 
ab   & 21.00 & 11.01 & 21.00 & 11.10 & \fbox{11.11} & 11.10 & 21.00 & 11.01 & 21.00\\
ac & 20.00 & 11.01 & 22.00& 10.10 & \fbox{11.11}& 12.10 & 20.00 & 11.01 & 22.00 \\ 
ba   & 10.10 & \fbox{11.11} & 10.10 & 10.10  & \fbox{11.11}& 10.10 & 10.10  & \fbox{11.11} & 10.10 \\ 
bb & 11.11 & \fbox{11.11} & 11.11 & \fbox{11.11} & \fbox{11.11} & 11.11 & 11.11 & \fbox{11.11} & 11.11 \\
bc & 10.10  &  \fbox{11.11}& 12.11& 10.10  & \fbox{11.11}&  \fbox{12.11}& 10.10  & \fbox{11.11}& 12.11\\
ca & 20.00 & 11.01 & 00.00  & 10.10 &  11.11& 10.10  & 00.20 & 01.21 & 00.20\\
cb &  21.00 & 11.01 & 01.01 & 11.10 & 11.11 & 01.11 & 01.20 & 01.21& 01.21   \\
cc & 20.00 & 11.01 & 02.02& 10.10 &  11.11 & 02.12 & 00.20 & 01.21 &  02.22
\end{array}\] 
\[
\begin{array}{l}
\begin{array}{|c|c|} \hline 1 & 2 \\ \hline a &  c\\ b & b\\ c & a  \\ \hline \end{array}\mbox{------}2\mbox{------}\begin{array}{|c|c|} \hline 1 & 2 \\ \hline c& c \\ b& b \\ a & a \\ \hline \end{array}\mbox{------}1\mbox{------}\begin{array}{|c|c|} 
	\hline 
	1 & 2 \\ 
	\hline 
	c & c \\ 
	b& b\\ 
	a & a \\ 
	\hline 
	\end{array}
\\
\hspace{.0cm} t, \profile \hspace{2.2cm} u, \profile' \hspace{2.2cm} v, \profile'
\end{array}
\hspace{.7cm}
\begin{array}{c|ccccccccc} 
1\backslash 2 & aa & ab & ac & ba & bb & bc & ca & cb & cc \\ 
\hline
aa & aaa & aab & aaa & bba &  bbb& bba & aaa & aab& aaa \\ 
ab & abb & abb & abb& bbb& bbb & bbb& abb & abb& abb\\
ac & aaa & aab & aac& bba & bbb& bbc  & aca & acb& acc\\ 
ba & baa & bab & baa & bba  & bbb& bba  & baa  &  bab & baa \\ 
bb & bbb & bbb & bbb & bbb & bbb & bbb & bbb & bbb & bbb \\
bc & baa & bab & bac& bba  & bbb & bbc& bca  & bcb & bcc \\
ca & aaa & aab & aaa  & bba & bbb & bba & caa & cab & caa \\
cb & abb & abb & abb & bbb & bbb & bbb& cbb & cbb & cbb \\
cc & aaa & aab & aac & bba & bbb & bbc & cca & ccb & ccc
\end{array}\] \[\begin{array}{c|ccccccccc} 
1\backslash 2 & aa & ab & ac & ba & bb & bc & ca & cb & cc  \\ 
\hline
aa   & 20.00 & 20.01& 20.00 & 10.10 & \fbox{11.11}& 10.10  & 20.00 &  20.01& 20.00 \\ 
ab   & 21.01 & 21.01 & 21.01 & \fbox{11.11}& \fbox{11.11} & \fbox{11.11}& 21.01 & 21.01 & 21.01\\
ac & 20.00 & 20.01 & 20.02& 10.10 & 11.11& \fbox{11.12} & 20.00 & 21.01 & 22.02\\ 
ba   & 10.00& 10.01 & 10.00 & 10.10  & \fbox{11.11}& 10.10 & 10.00  & 10.01 & 10.00 \\ 
						bb & 11.11 & 11.11 & 11.11 & \fbox{11.11} & \fbox{11.11} & \fbox{11.11} & 11.11 & 11.11 & 11.11 \\
						bc & 10.00 & 10.01 & 10.02& 10.10  & 11.11 & \fbox{11.12}& 10.10  &  11.11 & 12.12 \\
						ca & 20.00 & 20.01 & 20.00  & 10.10 & \fbox{11.11} & 10.10 & 00.00 & 00.01 & 00.00 \\
						cb & 21.01 & 21.01 & 21.01 & \fbox{11.11}& \fbox{11.11} & \fbox{11.11}& 01.11 & 01.11 & 01.11\\
						cc & 20.00 & 20.01 & 20.02 & 10.10 & 11.11 & 11.12 & 00.20 & 01.21 & 02.22 						\end{array}\]
}
\caption{Conditional equilibria for profile models where two states have the same profile}
\label{fig.tuvexample}
\end{figure}

Figure \ref{fig.tuvexample} also gives an overview of the conditional equilibria for both profile models, including the matrices with winners in order to calculate the payoffs. As it may be confusing to see three winners but four payoff values let us explain once more the mechanics of conditional profiles and conditional equilibria. For example, take the $t,u,v$ model, with conditional profile $(ac, bc)$ that is an equilibrium, where in the winners matrix we find $bbc$ for that entry, and where in the payoff matrix we find $11.12$. Conditional profile $(ac,bc)$ is the conditional profile such that
\begin{itemize}
\item If 1 prefers $a$ (i.e., in state $t$) then she votes $a$, and if 1 prefers $c$ (i.e., in states $u,v$) then she votes $c$.
\item If 2 is uncertain whether 1 prefers $a$ (i.e., in states $t,u$) then he votes $b$, and if 2 knows that 1 prefers $c$ (i.e., in state $v$) then he votes $c$.
\end{itemize}
The winners in states $t,u,v$ of these conditional preferences are, respectively, $b,b,c$. If the state is $t$, then 1 votes $a$ and 2 votes $b$, so $b$ wins. If $u$, then 1 votes $c$ and 2 votes $b$, so $b$ wins. If $v$, then 1 votes $c$ and 2 votes $c$, so $c$ wins.

The payoff entry is $11.12$ because: for voter 1 in state $t$, the worst (and only) outcome is $b$ with value 1, for voter 1 in states $u,v$ the worst of $b$ and $c$, with values 1 and 2, is (also) 1; for voter 2 in states $t,u$ the worst of $b$ and $b$, both with value 1, is 1, and for voter 2 the worst and only outcome in state $v$ is $c$ with value 2.

Conditional profile $(ac,bc)$ is an equilibrium. For this, we have to check four virtual voters. For example, voter 1 in state $t$ cannot do better, because the first digit of the payoff entries for profiles $(bc,bc)$, $(cc,bc)$ is not greater than 1; voter 1 in class $\{u,v\}$ cannot do better: we check the second digit of the entries for conditional profiles $(aa,bc)$ and $(ab,bc)$. Voter 2 cannot do better in $t,u$, check the third digit in entries for $(ac,ac)$ and $(ac,cc)$; and voter 2 also cannot do better in $v$, in which case we check the fourth digit in the entries for $(ac,ba)$ and $(ac,bb)$. We are done.


In the $s,t,u$ case, it does not make a difference to voter 2 whether he knows voter 1's preferences. If the profile is $\profile$, voter 1 {\em knows} that voting for $a$ is dominant. On that assumption, voter 2 should vote $b$, such that $b$ wins. Indeed, in almost all equilibria (except $(bc,bc)$), $b$ wins and the payoff is 1 for both voters. Unlike for the two-state example, where in all equilibria voter 2 votes $b$, there are now equilibria wherein voter 2 does not vote $b$. However, these are not really interesting, as 1 votes $b$ in these, which is dominated by 1 voting $a$.

On the other hand, in the $t,u,v$ case, it makes a difference to voter 2 if 1 is uncertain or not. There are equilibria wherein both players vote $c$ in state $v$, namely $(ac,bc)$ and $(bc,bc)$. Whereas there is no equilibrium wherein both players vote $c$ in state $u$, even though  that would have been just as much in their interest. We can easily justify this result by our intuitions. If voter 2 is uncertain about 1's preferences, the worst-case avoiding strategy remains voting $b$. If voter 2 knows that 1's preferences are his own, even 1's uncertainty is not enough to make him change his vote. The same cannot be said for voter 1 in that state $v$, she has to weigh the odds against voter 2 playing safe and voting $b$ instead; but either way, voting $c$ then also gives her best result. So voter 1 should be indifferent between $b$ and $c$ when arguing from that worst-case scenario, and this is indeed the case: $(ab,bc)$ and $(bb,bc)$ are also equilibria.

\section{Revealing voting preferences} \label{sec.dyn}

We can extend the setting for the interaction of voting preferences and knowledge of the previous sections with operations wherein voters are informed of other's preferences, thus reducing (`updating') their uncertainty. An obvious choice for such updates is the {\em public announcement} \citep{plaza:1989} of propositions about profiles. A public announcement can be modelled as an operation $\pmodel_s \mapsto (\pmodel|T)_s$, where $T \subseteq \States$ is the denotation in $\pmodel$ of a proposition about profiles $\phi$, and $\pmodel|T$ means model restriction to subdomain $T$. In that case we also write $(\pmodel|\phi)_s$.

Given a knowledge profile $\pmodel_s$, the precondition for execution of the operation {\em public announcement of} $\phi$ (or {\em update with} $\phi$) is that $\phi$ is true in $\pmodel_s$, and the way to execute it is to restrict the model $\pmodel$ to all the states where $\phi$ is true. We can then investigate the truth of propositions about profiles in that model restriction. This allows us to evaluate propositions about profiles of shape `after update with $\phi$, $\psi$ (is true),' such as: `After voter 1 reveals her preference, voter 2 knows that he has a manipulation'. 

All this is embodied in the following definition (see Section \ref{sec.logic} for a formal version).

\begin{definition}[Updated knowledge profile] \label{def.updated}
Let $\pmodel_s$ be a knowledge profile, where $\pmodel = (S,\sim,\pi)$, and let $\phi$ be a proposition about profiles with denotation $S'\subseteq S$ and such that $s \in S'$. Then the {\em updated knowledge profile} $(\pmodel|\phi)_s$ is defined as $\pmodel|\phi = (S',\sim',\pi')$ where ${\sim_i'} = {\sim_i} \cap {(S'\times S')}$ and for all ${\profile}\in O(\Candidates)^\Agents$, $\pi'(\profile) = \pi(\profile) \cap S'$.

Let $\psi$ also be a proposition about profiles. In $\pmodel_s$ it is true that $\psi$ after update with $\phi$, iff whenever $\phi$ is true in $\pmodel_s$, $\psi$ is true in $(\pmodel|\phi)_s$.
\end{definition}

A public announcement is considered information coming from an outside source (for example, a central authority) and that is reliable. However, it is common in dynamic epistemic logic to model a public announcement $\phi$ by an agent $a$ (an inside source, so to speak, that is modelled in the system) as the public announcement of $K_a \phi$. This makes it possible to formalize that a voter reveals her preference to the other voters, as above. But given this identification, a voter revealing her preferences to me is indistinguishable in the logical analysis from someone else revealing to me that voter's preferences. Yet another situation is when a voter reveals her preferences {\em only} to another voter but not to all voters. This is not a public announcement but a {\em private} announcement. Formalizing this is quite doable but requires a more complex analysis. Finally, announcements can be {\em insincere}, or {\em lies}. This also requires a more complex dynamic epistemic analysis.

We proceed with some results. Clearly, manipulations and equilibrium profiles are preserved after update, as these only depend on the profile of the actual state, that is preserved after any (truthful) update. These are not epistemic notions. For those, we have that:
\begin{proposition} \label{proprop}
Knowledge of manipulation is preserved after update.
\end{proposition}
\begin{proof}
In any state of the information set of the voter knowing the manipulation, the profile of that state has a manipulation, by Def.~\ref{def.knowmanip}. This is a universal property that is preserved after update. This holds for {\em de re} as well as {\em de dicto} knowledge.
\end{proof}

\begin{proposition} \label{proprop1}
Dominant manipulation is not preserved after update.
\end{proposition}
\begin{proof} This is an existential property that may not be preserved, namely if the (existential) strictness requirement is only satisfied in states that are removed in the update. 
\end{proof}

\noindent To investigate how conditional equilibria evolve after updates we first define the update of a conditional profile.
\begin{definition}[Updated conditional profile]
Let profile model $\pmodel = (\States,\sim,\pi)$ and conditional profile $[\profile]$ be given, where $[\vote]_i: {\States\setminus\!\sim_i} \imp {O(\Candidates)}$ are conditional preferences. Let $\phi$ be a proposition about profiles such that $\pmodel|\phi = (\States',\sim',\pi')$. Then the {\em updated conditional profile} $[\profile']$ (w.r.t.\ $\pmodel|\phi$) consists of conditional preferences $[\vote']_i : {\States'\setminus\!\sim'_i} \imp {O(\Candidates)}$ defined as: for all $s \in \States'$, $[\vote']_i([s]_{\sim'_i}) = [\vote]_i([s]_{\sim_i})$. 
\end{definition}
A preference $[\vote]_i: {\States\setminus\!\sim_i} \imp {O(\Candidates)}$ may be affected in three different ways in an update:
\begin{itemize}
\item An information set for voter $i$ disappears, namely if none of its states satisfies the update $\phi$. The updated preference then has one less condition (the virtual voter $(i,[s]_{\sim_i})$ ceases to exist).
\item An information set for voter $i$ shrinks, because some states satisfy $\phi$ and others do not satisfy $\phi$. We then have that $[s]_{\sim'_i} \subset [s]_{\sim_i}$. This may affect the value for $i$ of preference $\vote_i$: states with minimal value may have been removed, namely if $\min_i \, F(\pi([s]_{\sim_i})) \prec_i \min_i \, F(\pi'([s]_{\sim'_i}))$. 
\item An information set for $i$ remains the same, because all of its states satisfy $\phi$.  The expected worst outcome for $i$ remains the same.
\end{itemize}
\begin{proposition}
Conditional equilibrium is not preserved after update.
\end{proposition}
\begin{proof}
An information update may affect an equilibrium as follows. The outcome for a virtual voter $(i,[s]_{\sim_i})$ casting vote $\vote_i$ in the equilibrium is the worst outcome (winner) for voter $i$ in information set $[s]_{\sim_i}$. This payoff value of that winner is at least as good as the worst outcome in information set $[s]_{\sim_i}$ for any other preference $\vote_i''$. This is fragile and not preserved after update. More precisely, we may have that (we recall that for all ${\profile'} \in \pi([s]_{\sim_i})$, ${\vote'_i} = {\vote_i}$):
\[\begin{array}{rlll}
& \min_i \, F(\pi([s]_{\sim_i})) & \succeq_i & \min_i \, F(\pi([s]_{\sim_i})_{-i},\vote_i'') \\
\text{whereas } \ \ & \min_i \, F(\pi'([s]_{\sim'_i})) & \prec_i & \min_i \, F(\pi'([s]_{\sim'_i})_{-i},\vote_i'')
\end{array}\]
where as before $\pi'$ and $\sim_i'$ are the valuation and accessibility relation in the updated profile model; and this requires a ${\profile'} \in \pi({[s]_{\sim_i}}\!\!\setminus\!{[s]_{\sim_i'}})$ such that: 
\[\begin{array}{rlll}
& F(\profile')  & \succ_i & F(\profile'_{-i},\vote_i'')
\end{array} \vspace{-1cm} \] 
\end{proof}

\begin{proposition}
(Not being a conditional equilibrium) is not preserved after update.
\end{proposition}
\begin{proof}
This requires a preference $\vote_i''$ for which $\min_i \, F(\pi([s]_{\sim_i}) \prec_i \min_i \, F(\pi([s]_{\sim_i})_{-i},\vote_i'')$ and $\min_i \, F(\pi'([s]_{\sim'_i})) \succeq_i \min_i \, F(\pi'([s]_{\sim'_i})_{-i},\vote_i'')$.
\end{proof}

So conditional equilibria can both disappear and appear after updates (see also Example \ref{ex.doubidou} below). It is unclear to us if there are general patterns here. However, additional strategic behaviour comes into the picture with these negative results. An update may consist of a voter revealing her voting preferences. It may be that this voter's sincere preference is not part of an equilibrium conditional profile, but that after this voter reveals her sincere preference, the updated conditional profile is an equilibrium. This interaction between {\em strategic voting} and {\em strategic communication} may be of interest. 

\begin{example} \label{ex.doubidou}

We recall the two-state profile model presented in Section \ref{sec.ex1} and displayed again below. We now add dynamics to this example: voter 1 informs voter 2 of her true preference. This is also displayed below.

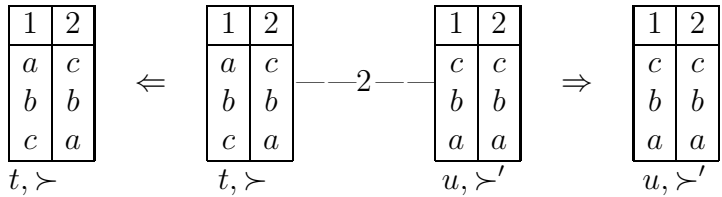
\[ \begin{array}{l}
\begin{array}{|c|c|} \hline 1 & 2 \\ \hline a & c \\ b & b \\ c & a \\ \hline \end{array}\quad \Pmi \quad\begin{array}{|c|c|} \hline 1 & 2 \\ \hline a & c \\ b & b \\ c & a \\ \hline \end{array}\mbox{------}2\mbox{------}\begin{array}{|c|c|} 
\hline 
1 & 2 \\ 
\hline 
c & c \\ 
b & b \\ 
a & a \\ 
\hline 
\end{array}
\quad \Imp \quad
\begin{array}{|c|c|} 
\hline 
1 & 2 \\ 
\hline 
c & c \\ 
b & b \\ 
a & a \\ 
\hline 
\end{array} \\
t, \profile \hspace{2cm} t, \profile \hspace{2.2cm} u, \profile'\hspace{1.7cm} u, \profile' \hspace{.2cm} \end{array} \]

In state $t$, after voter 1 informs voter 2 of her true preference $a \succ_1 c$, no uncertainty remains, and 1 and 2 commonly know that the profile is $\profile$. From Section \ref{sec.ex1} we further recall that equilibrium votes for $\profile$ are $(a,b),(b,b)$, and that conditional equilibria for the $t,u$ profile model have shape $(xy,b)$, where $x$ is voter 1's preference in $t$ and $y$ is voter 1's preference in $u$, and where all but $cc$ are equilibria. (See the payoff matrix in Figure \ref{fig.tuexample}.) 

We can now observe that all conditional equilibria are preserved after update. For example, given conditional profile $(bc,b)$, the updated conditional profile according to Def.~\ref{def.updated} is $(b,b)$. There is therefore no strategic incentive for voter 1 to inform voter 2 in state $t$.

On the other hand, in state $u$ voter 1 has an incentive to make her preference $c \succ_1 a$ known to 2. In the model with $\profile$ and $\profile'$, there is no equilibrium wherein 2 votes $c$. But after 1 informs 2 that her preferences are the same as his preferences, $(c,c)$ is an equilibrium. Most conditional equilibria of the two-state model are preserved, but not those where 1 votes $c$. For example, $(ac,b)$ was an equilibrium, but the updated profile $(c,b)$ is not an equilibrium.

Voter 1 has a strategic interest to reveal her preferences to voter 2 in state $u$ because subsequently she expects the outcome of the vote to be better. Before the update she expected $b$ to win, after update she expects $c$ to win, and $c \succ_1 b$. This demonstrates that: \begin{quote} {\em When there is uncertainty about votes, voters have different ways of acting strategically: voting strategically or strategically revealing voting preferences}. \end{quote}
\end{example}

\section{A logic of knowledge and voting} \label{sec.logic}


Given the set $\Agents$ of agents, the set of profiles $O(\Candidates)^\Agents$, that serve as {\em propositional variables}, and a voting rule $F$, we now define a logical language and semantics.
\begin{definition}[Logical language] \label{def.language}
The language $\lang$ is defined as
\[ \phi \ ::= \ {\profile} \mid \neg \phi \mid \phi \wedge \phi \mid K_i\phi \mid [\phi]\phi \]
where $i \in\Agents$ and ${\profile} \in O(\Candidates)^\Agents$. \end{definition}
An element of the language is a {\em formula}, and $\phi$ is a {\em formula variable}; $K_i \phi$ stands for `voter $i$ knows that $\phi$'; $[\phi]\psi$ stands for `after (public) announcement of $\phi$, $\psi$ (is true)'.

\begin{definition}[Semantics]
Let $\pmodel_s$ be a knowledge profile, where $\pmodel = (\States, \sim,\pi)$.The interpretation of formulas in a knowledge profile is defined as follows:
\[ \begin{array}{lcl}
\pmodel_s \models {\profile} & \text{ iff } \ & \pi(s) = {\profile} \\
\pmodel_s \models \neg\phi & \text{ iff } \ & \pmodel_s \not\models \phi \\
\pmodel_s \models \phi \wedge \psi & \text{ iff } \ & \pmodel_s \models \phi \text{ and } \pmodel_s \models \psi \\
\pmodel_s \models K_i \phi & \text{ iff } \ & \text{for every } t \text{ such that } s \sim_i t, \pmodel_t \models \phi \\
\pmodel_s \models [\phi]\psi & \text{ iff } \ & \pmodel_s \models \phi \text{ implies } (\pmodel|\phi)_s \models \psi
\end{array} \]
\end{definition}
where $\pmodel_s \not\models \phi$ stands for `not ($\pmodel_s \models \phi$)', and where $\pmodel|\phi = (S',\sim',\pi')$ such that $S' =
\{t \in S \mid \pmodel_t \models \phi\}$, $\sim_i'\ =\ \sim_i \cap \ {(S'\times S')}$, and for all $s \in S'$, $\pi'(s) = \pi(s)$. If $\pmodel_s \models \phi$ for all $s \in S$, we write $\pmodel \models \phi$ ($\phi$ is {\em valid on $\pmodel$}) and if this is the case for all $\pmodel$, we say that $\phi$ is {\em valid}, and we write $\models \phi$. \emph{The logic of knowledge and voting} is the set of validities for the class of profile models.

Although the voting function $F$ is not used in the syntax or in the semantics, it will later be used in the logical abbreviations introduced to formalize concepts such as conditional equilibrium.

Profile models are standard Kripke models, but with valuations that satisfy special conditions. We now present principles that are valid on the class of profile models, and that will feature as axioms in the proof system. Let $\overline{\Vel}$ denote exclusive disjunction.
\[ \begin{array}{llll}
\mathbf{P} &:& \overline{\Vel} \, O(\Candidates)^\Agents  \\
\mathbf{N} &:& \Et_{i \in \Agents} \Et_{\vote_i \in O(\Candidates)} ({\vote_i} \imp K_i {\vote_i})
\end{array}\]
Axiom $\mathbf{P}$ spells out that only one profile is true in each state. 
Axiom $\mathbf{N}$ says that voters know their own preferences. \begin{proposition} \label{prop.char}
The axioms $\mathbf{P}$
 and $\mathbf{N}$ are valid on the class of profile models.
\end{proposition}
\begin{proof}
Axiom $\mathbf{P}$ is valid as the valuation $\pi$ is a function from states to profiles. Axiom $\mathbf{N}$ formalizes the constraint on knowledge profiles that $s \sim_i t$ implies $\pi(s)_i = \pi(t)_i$.
\end{proof}


\begin{proposition} \label{prop.ax}
The logic of knowledge and voting has a complete axiomatization. 
\end{proposition}
\begin{proof}
As the axiomatization of the logic of knowledge and voting we propose the axiomatization of public announcement logic, that is standard and that we assume to be known \citep{plaza:1989}, and to which we add the axioms $\mathbf{P}$ and $\mathbf{N}$. This axiomatization is with respect to the class of profile models (see the definition of the logic, above). For the soundness we refer to the soundness of public announcement logic and the validity of $\mathbf{P}$ and $\mathbf{N}$ (Prop.~\ref{prop.char}). For the completeness we observe that the canonical model to determine the completeness of the logic without public announcements is a profile model (Prop.~\ref{prop.char}, again), and that the completeness of the logic with announcements is as usual (see \cite{hvdetal.del:2007}) obtained because every formula is equivalent to one without announcements (the axioms are rewriting rules, pushing all logical connectives beyond announcements, such as $[\phi](\psi\et\chi) \eq ([\phi]\psi\et[\phi]\chi)$).
\end{proof}


We proceed to formalize the concepts involving knowledge and voting. The main definition of interest is that of conditional equilibrium. We start with some abbreviations.\footnote{In the 4th abbreviation, note that the $\profile''$ such that $(\profile_{-i},\vote'_i) = {\profile''}$ is unique, and that $\Vel \{ {\profile''}\} = {\profile''}$.}
%
%
\[ \begin{array}{llll}
\vote_i & := & \Vel \{ {\profile'}\in O(\Candidates)^\Agents \mid {\vote_i} = {\vote'_i} \} \\ 
a \succ_i b & := & \Vel \{ {\profile'}\in O(\Candidates)^\Agents \mid a \vote'_i b  \} \\
a & := & \Vel \{ {\profile}\in O(\Candidates)^\Agents \mid F(\profile) = a \} \\
(\profile_{-i},\vote'_i) & := & \Vel \{ {\profile''}\in O(\Candidates)^\Agents \mid (\profile_{-i},\vote'_i) = {\profile''} \} \\
F(\profile') \succ_i F(\profile'') & := & (({\profile'} \imp a) \et ({\profile''} \imp b)) \imp (a \succ_i b) \\
\end{array} \]
We emphasize the role of the background parameter $F$ in these abbreviations. Firstly, the `formula' $a$ above (third item) stands for `candidate $a$ wins the election'. The definiens uses $F$. Formula $a$ is therefore true in a state $s$ with profile $\pi(s)$ such that $F(\pi(s)) = a$. Secondly, the use of $F$ in the definiendum of the final abbreviation (fifth item), $F(\profile') \succ_i F(\profile'')$, is therefore proper, because the definiens refers to these winners, namely in ${\profile'} \imp a$ and ${\profile''} \imp b$: that $a$ and $b$ are abbreviations of formulas that use $F$.

\begin{example}
We demonstrate the logic reusing Example \ref{ex.one} about Leela and Sunil. 
\begin{itemize}
\item Leela prefers Alice over Cars: $\pmodel_t \models a \succ_1 c$
\item Sunil does not know that Leela prefers Alice over Cars: $\pmodel_t \models \neg K_2 (a \succ_1 c)$
\item Leela knows Sunil's preference, but Leela is uncertain whether Sunil knows her preference: $\pmodel_t \models K_1 {\vote_2} \et \neg (K_1 K_2 {\vote_1} \vel K_1 \neg K_2 {\vote_1})$
\item Sunil does not know that Leela prefers Alice over Cars, but after he was told so, he knows it: $\pmodel_t \models \neg K_2 (a \succ_1 c) \et [a \succ_1 c] K_2 (a \succ_1 c)$
\end{itemize}
\end{example}

\noindent Using the above abbreviations and trivial ones like $a \succeq_i b := \neg (b \succ_i a)$ we now have that: 

\begin{itemize}
\item Voter $i$ has a manipulation of profile $\profile$: \\ $\Vel_{\vote_i'\in O(\Candidates)} (F(\profile_{-i},\vote'_i) \succ_i F(\profile))$
\item Voter $i$ has a manipulation $\vote'_i$ in profile $\profile$: \\ $F(\profile_{-i},\vote'_i) \succ_i F(\profile)$
\item Voter $i$ has a dominant manipulation $\vote'_i$ in profile $\profile$: \\ $\Et_{{\profile''}\in O(\Candidates)^\Agents} (F(\profile''_{-i},\vote'_i) \succeq_i F(\profile'')) \ \et \ \Vel_{{\profile''}\in O(\Candidates)^\Agents} (F(\profile''_{-i},\vote'_i) \succ_i F(\profile''))$
\item Voter $i$ with preference $\vote_i$ knows de dicto that she has a manipulation:\\  
$K_i \Vel_{{\vote'_i}\in O(\Candidates)} ({\profile''} \imp (F(\profile''_{-i},\vote'_i) \succ_i F(\profile'')))$
\item Voter $i$ with preference $\vote_i$ knows de re that she has a manipulation:\\ 
$\Vel_{{\vote'_i}\in O(\Candidates)} K_i ({\profile''} \imp (F(\profile''_{-i},\vote'_i) \succ_i F(\profile'')))$
\item Profile $\profile$ is an equilibrium profile: \\ $ \Et_{i \in \Agents} \Et_{\vote'_i\in O(\Candidates)} (F(\profile) \succeq_i F(\profile_{-i},\vote'_i))$
\item A conditional equilibrium (where $\mathcal{J}$ is defined below): \[ \Et_\mathcal{J} \left( (\Et_{i \in\Agents} \phi_i^{j(i)}) \imp  \Et_{k \in\Agents} \Et_{\vote'_k\in O(\Candidates)} \left( F(\Et_{i \in\Agents} \vote_i^{j(i)}) \ \ \succeq_k^{j(i)} \ \ F( \ (\Et_{i \in\Agents} \vote_i^{j(i)})_{-k},\vote'_k \ ) \right) \right) \]
\end{itemize} 

In the formalization of knowledge {\em de re} and {\em de dicto}, condition `${\profile''} \imp$' guarantees that only profiles $\profile''$ in the information set $\pi([t]_{\sim_i})$ of a knowledge profile $\pmodel_t$ are selected. 

The index set $\mathcal{J}$ in the definition of conditional equilibrium ranges over {\em distinguishing formulas} $\phi_i^{j(i)}$ for all ($n$) voters and for all information sets of those voters, i.e., $\mathcal{J} := \{ (j(1),\dots,j(n)) \mid 1 \leq j(1) \leq \max(1), \dots, 1 \leq j(n) \leq \max(n) \}$ where each voter $i$ has $\max(i)$ information sets; and where for condition $\phi_i^{j(i)}$ voter $i$ has preference $\vote_i^{j(i)}$. A distinguishing formula is only true in that information set of the voter and else false. As the domain is finite, such distinguishing formulas exist. The distinguishing formulas {\em cover} the profile model, i.e., let the information sets for a voter in a given profile model be numbered $1,\dots,\max$, then there are formulas $\phi^1,\dots,\phi^{\max}$ such that $\Vel_{j=1\dots \max} \phi^j$ and $\Et_{j=1\dots \max} (\phi^j \imp \neg\phi^{j+1})$ (where $\max+1 = 1$) are valid on the profile model. The formalization of conditional equilibrium then simply spells out the equilibrium for the Bayesian game with virtual agents (`virtual voters') $(i,[s]_{\sim_i})$ instead of voters $i$. Note that $\Et_{i \in\Agents} \vote_i^{j(i)}$ represents a profile: a validity of the logic is that ${\profile} \eq \Et_{i \in \Agents} \vote_i$.


\section{Conclusion and further research} \label{sec.concl}

We presented a logic for the interaction of voting and knowledge. The semantic primitive is the knowledge profile: a profile including uncertainty of voters about  the profile. We defined de re knowledge of manipulation and de dicto knowledge of manipulation, and the notion of conditional equilibrium for risk-averse voters. We modelled the dynamics of knowledge, such as voters revealing preferences, and its effects on knowledge of manipulation and on conditional equilibrium, where we proved that knowledge of manipulation is preserved after such updates whereas a conditional equilibrium may not be preserved. Finally, we formalized the system in a dynamic epistemic logic.

\medskip

Our goal was to present a \emph{minimally} meaningful extension of works on uncertainty in voting, namely that permits formalizing higher-order uncertainty and updates of uncertainty. Many further extensions are conceivable.  

Additionally to uncertainty over the preferences of other voters one can consider uncertainty over the voting function. Instead of profile models consisting of states with a profile for each such state, we would then need `voting models' consisting of states with a profile and a voting function for each state. Notions like pessimistic manipulation, conditional preference, and conditional equilibrium generalize straightforwardly to this setting. The logic would need an additional axiom to describe that only one voting function can be associated with a given state.

Apart from voters $1,\dots,n$ it is convenient to distinguish a designated additional agent 0 who is the {\em central authority}, or {\em chair}. The universal relation on a knowledge profile model can then be seen as the indistinguishability relation of that agent 0. This opens the door to the logical modelling of well-studied problems in computational social choice, such as control by the chair, or determining possible winners. 

Our results are for any amount of voters but all our examples were for at most three voters. This was for the purpose of presentation. Still, in realistic settings the power of individual voters is very limited. However, our results seem also to meaningfully apply to (few) \emph{coalitions} of voters. In voting theory, the power of a coalition means the power of a set of agents that can decide on a joint action as a result of communication between them. Under conditions of uncertainty about profiles this is therefore a function of the \emph{distributed knowledge} of that coalition. In a knowledge profile, the indistinguishability relation for a coalition is the intersection of the indistinguishability relations for the individuals in the coalition.

We modelled {\em knowledge} of preferences, not {\em belief}. Unlike knowledge, beliefs may be incorrect. Somewhat similarly, unlike truthful announcements, insincere announcements (e.g., lying about your true preferences in a voting poll) may lead to false beliefs. Conceptually, the interaction between belief and voting is much more complex than between knowledge and voting. Technically, there may be fewer issues, for example, the same logic as in Section \ref{sec.logic} can be used with minor adjustments.

\paragraph*{\bf Acknowledgements}

This work is based on Chapter 2 of the Ph.D.\ thesis of first author \citet{bakhtiari:2017}. It subsumes earlier work presented at AAMAS 2012 and TARK 2013  \citep{hvdetal.TARKvote:2013}.\footnote{The following results originate in the TARK version: Proposition 3 of TARK is essentially the same as our Proposition \ref{proprop1}, and Proposition 4 of TARK is exactly the same as our Proposition \ref{proprop}. Note that the formalization of voting concepts in Definition 13 of TARK is very different from the formalization in our Section \ref{sec.logic}, both in form and in meaning: the propositional letters in the TARK logic stand for `agent $i$ prefers candidate $a$ over candidate $b$', whereas our proposition letters stand for profiles.} We thank anonymous conference reviewers and anonymous journal reviewers, and the members of Zeinab Bakhtiari's Ph.D.\ committee, for their comments. We are indebted to J\'er\^ome Lang, CNRS --- Universit\'e Paris Dauphine, for his unfailing support of and interest in our work, and to Helle Hansen, Delft University of Technology, for her detailed proofreading and utterly precise corrections. Hans van Ditmarsch is also affiliated to IMSc, Chennai, India. We acknowledge support from ERC project EPS 313360.

\bibliographystyle{natbib}
\bibliography{vote,biblio2019,biblio2019b}

\end{document}